\documentclass[12pt,onecolumn,journal]{IEEEtran}

\usepackage{amsfonts}
\usepackage{setspace}
\usepackage{clrscode}
\usepackage{amsthm}
\usepackage{pict2e}

\usepackage{amsmath}
\usepackage{amssymb}
\usepackage[dvips]{graphicx}
\usepackage{epsfig}
\usepackage{hyperref}
\usepackage{algorithm}
\usepackage{algorithmic}
\usepackage{caption}
\usepackage{subcaption}

\newcommand{\rme}{{\mathrm{e}}}
\newcommand{\E}{\mathbb{E}}
\newcommand{\EE}{\mathcal{E}}

\newcommand{\figsize}{0.5}

\newcommand{\f}{\text{F}^{-1}}

\newcommand{\avg}{\text{avg}}

\newcommand{\R}{R^\ast}

\newcommand{\ssnr}{\text{\scriptsize{SNR}}}

\newcommand{\tSNR}{\text{\footnotesize{SNR}}}

\makeatletter

\newcommand{\Rmnum}[1]{\expandafter\@slowromancap\romannumeral #1@}
\makeatother

\newcommand{\y}{\mathbf{y}}
\newcommand{\x}{\mathbf{x}}
\newcommand{\n}{\mathbf{n}}

\newcommand{\ben}{\frac{E_b}{N_0}}
\newcommand{\tsnr}{{\text{\footnotesize{SNR}}}}
\newcommand{\db}{\text{dB}}

\newtheorem{Lem}{Theorem}
\newtheorem{Rem}{Remark}

\allowdisplaybreaks
\begin{document}

\begin{spacing}{1.6}

%
\title{Energy Efficiency of Hybrid-ARQ under Statistical Queuing Constraints}

\author{Yi Li, Gozde Ozcan, M. Cenk Gursoy and Senem
Velipasalar
\thanks{The authors are with the Department of Electrical
Engineering and Computer Science, Syracuse University, Syracuse, NY, 13244
(e-mail: yli33@syr.edu, gozcan@syr.edu, mcgursoy@syr.edu, svelipas@syr.edu).}
\thanks{The material in this paper was presented in part at the Annual Conference on Information Sciences and Systems (CISS), Princeton University, Princeton, NJ,  in March 2014.}
}

\maketitle

\begin{abstract}
In this paper, energy efficiency of hybrid automatic repeat request (HARQ) schemes with statistical queuing constraints is studied for both constant-rate and random Markov arrivals by characterizing the minimum energy per bit and wideband slope. The energy efficiency is investigated when either an outage constraint is imposed and (the transmission rate is selected accordingly) or the transmission rate is optimized to maximize the throughput. In both cases, it is also assumed that there is a limitation on the number of retransmissions due to deadline constraints. Under these assumptions, closed-form expressions are obtained for the minimum energy per bit and wideband slope for HARQ with chase combining (CC). Through numerical results, the performances of HARQ-CC and HARQ with incremental redundancy (IR) are compared. Moreover, the impact of source variations/burstiness, deadline constraints, outage probability, queuing constraints on the energy efficiency is analyzed.
\end{abstract}
\thispagestyle{empty}



\section{Introduction}
 In wireless communications, increasing transmission rates, improving energy efficiency, reducing delays, and guaranteeing reliable and robust data transmission are key considerations with often contradictory requirements in terms of the use of limited resources. For instance, increasing rates, reducing delays, establishing robust communication links may lead to increased energy consumption, hurting energy efficiency. Moreover, due to the influence of noise, fading, multipath propagation and Doppler frequency shift, the performance of wireless systems is highly sensitive to mobility and changes in the environment. The automatic repeat request (ARQ) and forward error correction (FEC) are two kinds of widely used schemes applied in order to ensure reliable delivery of data in such challenging wireless channel conditions. While ARQ facilitates the retransmission of erroneously received data packets with feedback from the receiver to the transmitter, FEC schemes enable the correction of transmission errors without retransmission by adding redundancy to the data. In order to provide better error correction performance and lower implementation cost, ARQ and FEC schemes are combined to develop hybrid ARQ (HARQ) \cite{wicker1995error}. HARQ protocol has the ability to adapt the transmission rate to time-varying channel conditions with limited channel side information (CSI) at the transmitter. In HARQ with chase combining (HARQ-CC) and HARQ with incremental redundancy (HARQ-IR) schemes, the corrupted packets are not deleted but rather stored and combined in the next transmission. A very detailed study on the performance of HARQ-CC and HARQ-IR protocols was provided in \cite{HARQ_Gaussian}, in which the throughput was characterized following an outage probability analysis. Also, the throughput analysis of HARQ-CC and HARQ-IR schemes subject to an outage constraint has been conducted in \cite{fix_outage}. In addition to reliability, energy efficiency is another concern in wireless communications, due to limited battery power in mobile systems, growing energy demand, and high energy costs as well as environmental concerns. The energy efficiency of HARQ protocols has been addressed recently. For instance, the energy efficiency of HARQ-CC and HARQ-IR schemes for delay insensitive systems was studied in \cite{stanojev}.

In addition, many wireless applications require certain quality-of-service (QoS) guarantees for acceptable performance levels at the end-user, especially in delay sensitive scenarios, such as live video transmission, interactive video (e.g., teleconferencing), and mobile online gaming. In such cases, effective capacity can be employed to characterize the system throughput under statistical queuing constraints \cite{dapeng}, which require the buffer overflow probabilities to decay exponentially fast asymptotically as the buffer threshold grows without bound. In the presence of such QoS constraints, it is critical to evaluate the performance of HARQ schemes since they involve retransmissions. With this motivation, the authors in \cite{choi} analyzed the impact of different power allocation schemes on energy per bit and effective transmission delay of HARQ-IR in a multiuser downlink channel. Moreover, the recent work in \cite{choi2} mainly focused on the performance comparison between adaptive modulation and coding (AMC) and HARQ-IR in terms of energy efficiency under QoS constraints. The authors considered the notion of effective capacity and applied it to AMC. The performance of HARQ-IR was analyzed under a QoS constraint described in terms of packet loss probabilities. Recently, we in \cite{HARQ} employed the effective capacity formulation and provided a characterization of the effective capacity of HARQ under statistical queuing constraints.

In effective capacity analysis, constant-rate arrivals are assumed at the transmitter. On the other hand, randomly time-varying arrivals are frequent in real applications. For instance, the data traffic can be regarded as an ON-OFF process in voice communications (e.g., in VoIP) and variable bit-rate video traffic is statistically characterized as autoregressive, Markovian, or Markov-modulated processes \cite{survey-VBRvideotraffic}. With this motivation, the authors in \cite{SB_QoS1} studied the impact of source burstiness on the energy efficiency under statistical queuing constraints, and they further developed energy-efficient power control policies in \cite{SB_QoS2} considering Markov arrivals.


In this paper, we study the energy efficiency of HARQ under statistical queuing constraints in the low power and low QoS exponent regimes for both constant-rate and random arrival models. More specifically, our contributions are the following:
\begin{enumerate}
  \item We characterize the throughput of HARQ-CC and then derive closed-form minimum energy per bit and wideband slope expressions in the presence of statistical QoS constraints while satisfying a target outage probability.
  \item Our initial analysis addresses constant-rate arrivals\footnote{These early results were also reported in the conference version \cite{EE_HARQ} of our paper.}. Subsequently, we extend our analysis to random arrival models. More specifically, we consider ON-OFF discrete Markov and Markov fluid sources, and ON-OFF Markov modulated Poisson sources (MMPS). Analytical characterization are obtained for any type of channel fading (while numerical results consider Rayleigh and Nakagami fading.)
  \item We identify the impact of random arrivals and source burstiness on the energy efficiency of HARQ systems under statistical QoS constraints.
  \item Following our results for a given fixed outage probability, we determine the energy efficiency when throughput-maximizing transmission rates are employed.
\end{enumerate}

The remainder of the paper is organized as follows. In Section \ref{sec:system_model}, we describe the system model and the operational characteristics of the HARQ schemes. Preliminary concepts and formulations regarding statistical queuing constraints, throughput, and energy efficiency are introduced in Section \ref{sec:prep}. Energy efficiency of HARQ-CC is studied in detail in Section \ref{sec:EE_fixed_outage} for both constant-rate and random arrival models. In Section \ref{sec:energy-eff-optimal-rate}, we investigate energy efficiency with optimal transmission rates. Finally, numerical results are given in Section \ref{sec:numerical} and the paper is concluded in Section \ref{sec:conc}. Proofs are relegated to the Appendix.

\section{System Model} \label{sec:system_model}

\begin{figure}
\center
\includegraphics[width=\figsize\textwidth]{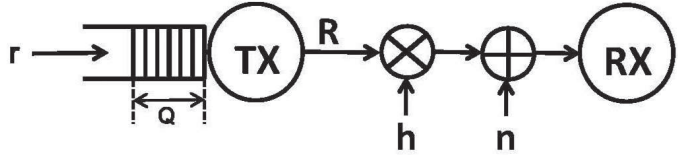}
\caption{System Model}\label{fig:model}
\end{figure}

In this paper, as depicted in Figure \ref{fig:model}, we consider a point-to-point wireless communication system, in which data packets arriving from the source are initially stored in a buffer at the transmitter before being sent over a fading channel to a receiver. We assume a block flat-fading model in which the fading coefficients stay the same within one block, but change independently across blocks. Each fading block is assumed to have a duration of $m$ symbols. Throughout the paper, we use subscript $i$ as the discrete time index. Under these assumptions, the received signal in the $i^{th}$ block can be written as
\begin{align}
\y_i=h_i \x_i+\n_i \hspace{0.5cm} i= 1,2, \dots
\end{align}
Above, $\x_i$  and $\y_i$  are the transmitted and received signal vectors of length $m$, respectively, and $h_i$ denotes the channel fading coefficient in the $i^{th}$ block. Also, $\n_i$ represents the noise vector with independent and identically distributed (i.i.d.) circularly-symmetric, zero-mean Gaussian components, each with variance $N_0$. Then, the instantaneous capacity (bits/s/Hz) in the $i^{\text{th}}$ block is given by
\begin{align}
C_i= \log_2(1+\tSNR z_i),
\end{align}
where $z_i=|h_i|^2$ is the magnitude-square of the fading coefficient, $\tSNR =\frac{\EE}{N_0}$ denotes signal-to-noise ratio, and $\EE$ is the average energy of each component of the  transmitted vector $\x_i$.

To guarantee the reliability of the system, we assume that the system employs HARQ scheme with fixed transmission rate $R$ (bit/s/Hz).
If the receiver decodes the received packet correctly, it sends an acknowledgment (ACK) feedback to the transmitter through an error free feedback link, and a new packet will be sent in the next time block. If the receiver cannot decode the packet, a retransmission request is sent through the feedback link, and another codeword block of the same packet will be sent in the next time block. Retransmission continues until the receiver gets the packet without error or if the limit on the number of retransmissions is reached, and then the corresponding packet will be removed from the buffer at the transmitter.

In the HARQ-IR scheme, additional information is sent in each retransmission and the receiver combines and decodes after the $M^{\text{th}}$ round of retransmissions without error only if $R$ satisfies
\begin{align}
R \leq \sum_{i=1}^M\log_2(1+\tSNR z_i).\label{r_HARQTR}
\end{align}

In the HARQ-CC scheme, the same coded data is transmitted in each retransmission. The receiver employs maximum-ratio-combining and decodes the data packet error-free after the $M^{\text{th}}$ round only if $R$ satisfies
\begin{align}
R \leq \log_2\left(1+\tSNR\sum_{i=1}^M z_i\right).\label{r_HARQCC}
\end{align}

Although the transmitter always sends information at a fixed rate, HARQ protocol effectively leads to rate adaptation depending on when the data is correctly decoded at the receiver. For instance, when the channel conditions are favorable, the transmission of a single packet can be completed within a few blocks, resulting in a relatively large average transmission rate, and vice versa if the channel conditions are poor. If the transmission of a single packet is completed in $N$ fading blocks, then from (\ref{r_HARQTR}) and (\ref{r_HARQCC}), one can easily show that the average transmission rate $R/N$ is bounded as,
\begin{align}\label{eq:R_bound}
\begin{cases}
\frac{1}{N}\sum_{i=1}^{N-1}\log_2(1+\tSNR z_i)< \frac{R}{N} \\
\hspace{1.5cm}\leq \frac{1}{N}\sum_{i=1}^N\log_2(1+\tSNR z_i) \hspace{1cm}\text{for HARQ-IR} \\
\frac{1}{N}\log_2\left(1+\tSNR\sum_{i=1}^{N-1} z_i\right)< \frac{R}{N} \\
\hspace{1.5cm}\leq \frac{1}{N}\log_2\left(1+\tSNR\sum_{i=1}^N z_i\right) \hspace{0.7cm}\text{for HARQ-CC},
\end{cases}
\end{align}
which implies that both HARQ-IR and HARQ-CC have the ability to adapt the average transmission rate to the channel conditions without requiring perfect channel side information (CSI) at the transmitter.

\section{Preliminaries}\label{sec:prep}

\subsection{Statistical Queuing Constraints and System Throughput}
Throughout this paper, we assume that the transmitter is operating under a queuing constraint, which requires the buffer overflow probability to decay exponentially fast, i.e.,
\begin{align}
\Pr\{Q \ge q\} \approx e^{-\theta q}, \label{eq:bufferconstraint}
\end{align}
for sufficiently large $q$, where $Q$ is the stationary queue length, $q$ is the overflow threshold, and $\theta$ is called the QoS exponent. More rigorously, QoS exponent $\theta$ is defined as
\begin{align}
\theta=\lim_{q \to \infty} \frac{-\log \Pr\{Q \ge q\}}{q}.
\end{align}
Note that $\theta$ is a factor that controls the exponential decay rate of the buffer overflow probability.
Indeed, a closer approximation for the overflow probability is given by \cite{dapeng}
\begin{gather}
\Pr\{Q \ge q\} \approx \varsigma e^{-\theta q} \label{eq:overflowprob-rev}
\end{gather}
where $\varsigma = \Pr\{Q > 0\}$ is the probability of non-empty buffer.
From (\ref{eq:overflowprob-rev}), we notice that higher values of $\theta$ indicate stricter limitations on the buffer overflow probability, leading to more stringent QoS constraints whereas lower values of $\theta$ represent looser QoS requirements. Conversely, for a given buffer threshold $q$ and overflow probability limit $\epsilon = \Pr\{Q \ge q\}$, the desired value of $\theta$ can be determined as
\begin{gather}
\theta = \frac{1}{q}\log_e \frac{\varsigma}{\epsilon}.
\end{gather}

The system throughput can be characterized as the maximum average arrival rate $r_{\text{avg}}$ that can be supported under statistical queuing constraints, described by (\ref{eq:overflowprob-rev}). According to the effective bandwidth and effective capacity formulations provided in \cite{chang} and \cite{dapeng}, respectively, in the presence of queuing constraints with QoS exponent $\theta$, the arrival process $a_i$ and departure process $c_i$ at the buffer should satisfy
\begin{align}
\Lambda_a(\theta)+\Lambda_c(-\theta)=0,\label{eq:qos}
\end{align}
where $\Lambda_p(\theta)=\lim_{t\to \infty}\frac{1}{t}\log_e\;\E\{e^{\theta\sum_{i=1}^{t}p_i}\}$ is the asymptotic logarithmic moment generating function (LMGF) of the random process $p_i$.

When the arrival rate is constant i.e., $a_i = a$ for all $i$, it can be easily seen that
\begin{align}
\Lambda_a(\theta) = a \theta.
\end{align}
Then, from (\ref{eq:qos}), we have
\begin{align}
a = -\frac{1}{\theta} \Lambda_c(-\theta). \label{eq:eff-cap-formula-constant-rate}
\end{align}\
Indeed, the right-hand side of (\ref{eq:eff-cap-formula-constant-rate}) is defined as the effective capacity of the wireless link \cite{dapeng}
\begin{align}
C_E(\theta, \tSNR)=-\frac{1}{\theta}\Lambda_c(-\theta),
\end{align}
characterizing the maximum constant arrival rate that can be supported by the time-varying wireless transmission rates while satisfying the statistical queueing constraint in (\ref{eq:overflowprob-rev}). Notice that under the constant-rate arrival assumption, the system throughput (or equivalently the maximum average arrival rate) is also given by the effective capacity:
\begin{equation}\label{eq:eff-cap-def}
r_{\text{avg}}(\theta,\tsnr)=\E\{a_i\}= a = C_E(\theta, \tSNR)=-\frac{1}{\theta}\Lambda_c(-\theta).
\end{equation}
In \cite{HARQ}, the effective capacity of HARQ-CC and HARQ-IR with fixed transmission rate is studied, and the following closed-form approximate expression is determined for small $\theta$ :
\begin{align}\label{eq:CE_HARQ}
r_{\text{avg}}(\theta,\tsnr) = C_E(\theta, \tSNR)= \frac{R}{\mu_1}-\frac{R^2 \sigma^2}{2\mu_1^3}\theta + o(\theta),
\end{align}
where $R$ denotes fixed transmission rate, $\mu_1$ and $\sigma^2$ are the mean and variance of $\hat{T}$, the total duration of time that has taken to successfully send one message.

When the arrival rate is not constant, the computation of the system throughput is more complicated. In general, we need to formulate the LMGF of the arrival process as a function of the average arrival rate, and obtain the throughput by solving (\ref{eq:qos}).

\subsection{Energy Efficiency Metrics}
As mentioned in the previous subsection, the system throughput is characterized by the average arrival rate $r_{\text{avg}}$. Moreover, we choose energy per bit, defined as
\begin{equation}
 \ben=\frac{\tSNR}{r_{\text{avg}}(\theta, \tSNR)},
\end{equation}
as the metric for energy efficiency under statistical QoS constraints.

In the low-$\tSNR$ regime, the throughput curve is characterized by the minimum energy per bit and the wideband slope \cite{verdu}. The minimum energy per bit is obtained from
\begin{equation}\label{bitenergy}
  \ben_{\rm{min}}=\lim_{\tsnr\rightarrow 0}\frac{\tsnr}{r_{\text{avg}}(\theta, \tsnr)}=\frac{1}{\dot{r}_{\text{avg}}(\theta, 0)}
\end{equation}
where $\dot{r}_{\text{avg}}(\theta,0)$ denotes the first derivative of the system throughput $r_{\text{avg}}(\theta, \tSNR)$ with respect to $\tSNR$ at zero $\tSNR$. Correspondingly, the wideband slope is the slope of the throughput curve at $ \ben_{\rm{min}}$ and is given by
\begin{equation}\label{widebandslop}
  S_0= \frac{-2(\dot{r}_{\text{avg}}(\theta,0))^2}{\ddot{r}_{\text{avg}}(\theta,0)}\log_e 2.
\end{equation}
Above, $\ddot{r}_{\text{avg}}(\theta,0)$ denotes the second derivative of $r_{\text{avg}}(\theta, \tSNR)$\footnote{In the remainder of the paper, especially when $\theta$ is fixed and derivatives with respect to $\ssnr$ are considered, we generally express average arrival rate and effective capacity only as a function of $\ssnr$ explicitly as $r_{\avg}(\ssnr)$ and $C_E(\ssnr)$, respectively, and suppress $\theta$ in order to avoid cumbersome expressions.} with respect to $\tSNR$ at zero $\tSNR$. Then, the throughput can be approximated as
\begin{equation}\label{sp-eff-curve}
 r_{\text{avg}}=\frac{S_0}{10\log_{10}2}\left( \ben_{\db}-\ben_{\rm{min},\db}\right)+\epsilon,
\end{equation}
where $\ben_{\db}=10 \log_{10} \ben$, and $\epsilon=o\left(\ben-\ben_{\rm{min}}\right)$. Hence, $\ben_{\min}$ and $S_0$ provide a linear approximation of the $r_{\avg}$ vs. $\ben$ curve in the vicinity of $\ben_{\min}$.

\section{Energy Efficiency of HARQ-CC scheme with Fixed Outage Probability}\label{sec:EE_fixed_outage}
In this section, we study the energy efficiency of HARQ-CC scheme with fixed outage probability. Initially, we consider constant-rate arrivals, characterize throughput by employing the effective capacity formulation, and derive the minimum energy per bit and wideband slope Subsequently, we incorporate random arrival models by considering discrete-time Markov, Markov fluid, and Markov modulated Poisson sources and determine the system throughput and analyze the energy efficiency again by determining the minimum energy per bit and wideband slope.

\subsection{Energy Efficiency of HARQ-CC with Constant Arrivals}
Before obtaining the minimum energy per bit and wideband slope expressions for HARQ-CC, we first characterize the system throughput of HARQ-CC scheme subject to an outage constraint. An outage event happens if the receiver does not correctly decode the message at the end of the $M^{\text{th}}$ HARQ round. More specifically, the outage probability is expressed as
\begin{equation}\label{eq:outage}
P_{\text{out}} = \Pr\left\{\log_2\left(1+\tSNR\sum_{i=1}^M z_i\right) < R\right\} = \varepsilon,
\end{equation}
where $M$ denotes the limit on the maximum number of HARQ rounds, reflecting the deadline constraint. Correspondingly, the transmission rate that guarantees an outage probability of $\epsilon$ can be expressed as
\begin{equation}\label{r-chase}
 R=\log_{2} \left(1+\f_M(\varepsilon)\tsnr\right),
\end{equation}
where $\f_M$ is the inverse cumulative distribution function (CDF) of $\sum_{i=1}^{M}z_i$. Specifically, for Rayleigh fading, $\frac{2}{\E\{z\}}\sum_{i=1}^{M}z_i$ follows a chi-square distribution with $2M$ degrees of freedom; for Nakagami-$m$ fading, $\sum_{i=1}^{M}z_i$ follows a Gamma distribution with shape parameter $Mm$ and scale parameter $\E\{z\}/m$.

Hence, using the above rate expression and the formulation in (\ref{eq:CE_HARQ}), we can express, for small $\theta$, the throughput of the HARQ-CC scheme subject to an outage constraint $\epsilon$ as
\begin{align}\label{eq:CE_HARQ_CC}
r_{\text{avg}}(\tsnr)=&\frac{\log_{2}(1+\f_M(\varepsilon)\;\tsnr)}{\mu} \notag \\
             &\hspace{1cm}-\frac{\left[\log_{2}(1+\f_M(\varepsilon)\;\tsnr)\right]^2\sigma^2\theta}{2\mu^3}.
\end{align}
In order to obtain the expressions of $\mu$ and $\sigma^2$, we first write the probability $P\{\hat{T}=kM+t\}$ that the transmission of the first $k$ messages have ended in failure due to the violation of the deadline constraint $M$, and the $(k + 1)^{\text{th}}$ message is successfully transmitted after $t \le M$ HARQ rounds as follows:
\begin{align}
\Pr\{\hat{T}=kM+t\}= (\Pr\{T>M\})^k \Pr\{T=t\} \label{eq:hatTdistribution}
\end{align}
where $\hat{T}$ denotes the total duration of time spent for successful message transmission, which includes failed transmissions due to the deadline constraint, and $T$ represents the random transmission time of each message. Above, $\Pr\{T>M\}$ is equal to the outage probability of  $\epsilon$, and $\Pr\{T=t\}$ can be expressed as
\begin{align}
\Pr\{T=t\} &= \Pr\{T\le t\}-\Pr\{T\le t-1\}\\
           &=\Pr\Bigg\{\log_{2}\Big(1+\tsnr\sum^{t}_{i=1}z_i\Big)\geqslant R\Bigg\} \notag\\
           &\hspace{2cm}-\Pr\Bigg\{\log_{2}\Big(1+\tsnr\sum^{t-1}_{i=1}z_i\Big)\geqslant R\Bigg\} \\
          &=\Pr\left\{\sum^{t}_{i=1}z_i\geqslant \f_M(\varepsilon)\right\} -\Pr\left\{\sum^{t-1}_{i=1}z_i\geqslant \f_M(\varepsilon)\right\}    \\
          &=\text{F}_{t-1}\left(\f_M(\varepsilon)\right)-\text{F}_{t}\left(\f_M(\varepsilon)\right)
\end{align}
where $\text{F}_{t}$ is the CDF of $\sum_{i=1}^t z_i$.
Now, (\ref{eq:hatTdistribution}) can be expressed as
\begin{align}
\Pr\{\hat{T}=kM+t\} =\varepsilon^k\;\left(\text{F}_{t-1}\left(\f_M(\varepsilon)\right)-\text{F}_{t}\left(\f_M(\varepsilon)\right)\right).
\end{align}
Having determined the distribution of $\hat{T}$, we can express the expected value and variance of $\hat{T}$. The expected value $\E\{\hat{T}\}=\mu$ can be found as
\begin{align}
   \mu&=\sum^{\infty}_{\hat{t}=1}\hat{t}\;\Pr\{\hat{T}=\hat{t}\} \label{eq:mu_0} \\
      &=\sum^{M}_{t=1}\sum^{\infty}_{k=0}(kM+t)\;\Pr\{\hat{T}=kM+t\} \label{eq:mu_1} \\
       &=\sum^{M}_{t=1}\left(\sum^{\infty}_{k=0} (kM+t)\varepsilon^k\;\Pr\{T=t\}\right)  \\
       &=\sum^{M}_{t=1}\left(t\;\Pr\{T=t\}\sum^{\infty}_{k=0}\varepsilon^k+M\;\Pr\{T=t\}\sum^{\infty}_{k=0}k \varepsilon^k\right) \label{eq:mu_2}\\
       &=\frac{1}{1-\varepsilon}\sum^{M}_{t=1}t\;\Pr\{T=t\}+ \frac{M\varepsilon}{(1-\varepsilon)^2}\sum^{M}_{t=1}\Pr\{T=t\}\label{eq:mu_3} \\
       &=\frac{1}{1-\varepsilon}\sum^{M}_{t=1}t\;\Pr\{T=t\}+ \frac{M\varepsilon}{1-\varepsilon}.\label{eq:mu_4}
\end{align}
Above, in (\ref{eq:mu_1}), we replace $\hat{t}$ by $kM+t$ and sum over both $k$ and $t$ in order to more explicitly address possible violations of maximum retransmission limit before successful packet transmission. Noting that $\sum^{\infty}_{k=0}\varepsilon^k=\frac{1}{1-\varepsilon}$ and $\sum^{\infty}_{k=0}k\varepsilon^k=\frac{\varepsilon}{1-\varepsilon}$, (\ref{eq:mu_2}) can be simplified to (\ref{eq:mu_3}). Notice that ${\sum^{M}_{t=1}\Pr\{T=t\}} = \Pr\{T \leqslant M\}$ represents the probability that the transmission has been completed before violating the deadline constraint $M$, and hence is equal to $1-\varepsilon$. Applying this fact to (\ref{eq:mu_3}), we obtain (\ref{eq:mu_4}).

Similarly, the variance of $\hat{T}$ is given by
\begin{equation}
  \sigma^2=\E\{\hat{T}^2\}-\mu^2
\end{equation}
where
\begin{align}
&\E\{\hat{T}^2\}=\sum^{\infty}_{\hat{t}=1}\hat{t}^2\;\Pr\{\hat{T}=\hat{t}\} \\
&=\sum^{M}_{t=1}\left(\sum^{\infty}_{k=0}\varepsilon^k (kM+t)^2\;\Pr\{T=t\}\right)\label{eq:sgm_1}\\
 &=\frac{1}{1-\varepsilon}\sum^{M}_{t=1}t^2\;\Pr\{T=t\}
                  +\frac{2M\varepsilon}{(1-\varepsilon)^2}\sum^{M}_{t=1}t\;\Pr\{T=t\}\notag\\
                  &\hspace{5.5cm}+\frac{M^2\varepsilon(1+\varepsilon)}{(1-\varepsilon)^2}\label{eq:sgm_2}.
\end{align}

Akin to the steps applied from (\ref{eq:mu_0}) to (\ref{eq:mu_4}), we again sum over $kM+t$ in (\ref{eq:sgm_1}), and then compute several summation terms with respect to $k$. Subsequently, using the fact that $\sum^{M}_{t=1}\Pr\{T=t\}=1-\varepsilon$, we obtain (\ref{eq:sgm_2}).

\begin{Rem}\label{rem1}
For Rayleigh fading, the expressions above can further be simplified using the relationship between the Poisson distribution and chi-square distribution \cite{poisson}. More specifically, the retransmission time $T-1$ follows a Poisson distribution and hence we have
\begin{equation}
  \Pr\{T=t\}=\frac{\lambda^{(t-1)}}{(t-1)!} \rme^{-\lambda} \label{eq:Poissondistribution}
\end{equation}
where $\lambda=\frac{1}{\E\{z\}}\f_{M}(\varepsilon)$. Inserting (\ref{eq:Poissondistribution}) into (\ref{eq:hatTdistribution}), we derive $\Pr\{\hat{T}=kM+t\}$ as
\begin{align}
\Pr\{\hat{T}=kM+t\} =\varepsilon^k\;\frac{\lambda^{(t-1)}}{(t-1)!} \rme^{-\lambda}.
\end{align}
Hence, inserting (\ref{eq:Poissondistribution}) into the expressions of $\mu$ and $\sigma^2$, we can further simplify their expressions as follows:
\begin{align}\label{eq:expected_T}
\mu=\frac{1}{1-\varepsilon}\sum^{M}_{t=1}\;\frac{t\lambda^{(t-1)}}{(t-1)!}\rme^{-\lambda} + \frac{M\varepsilon}{1-\varepsilon},
\end{align}
\begin{align}
\begin{split}\label{eq:sigma_exp}
\sigma^2\!=\!\frac{1}{1-\varepsilon}\sum^{M}_{t=1}\;\frac{t^2\lambda^{(t-1)}}{(t-1)!} \rme^{-\lambda}-\frac{1}{(1-\varepsilon)^2}&\Bigg(\sum^{M}_{t=1}\;\frac{t\lambda^{(t-1)}}{(t-1)!} \rme^{-\lambda}\!\Bigg)^2\\&\hspace{1.2cm}+\frac{M^2\varepsilon}{(1-\varepsilon)^2}.
\end{split}
\end{align}
\end{Rem}

Note that the expressions of $\mu$ and $\sigma^2$ do not depend on $\tSNR$. In the following result, we characterize the energy efficiency in the low $\tsnr$ regime for small $\theta$.
\begin{Lem} \label{theo:discrete}
For small QoS exponent $\theta$, the minimum energy per bit and wideband slope of the HARQ-CC scheme with the outage constraint $\epsilon$ are given, respectively, by
\begin{align} \label{eq:min_Eb}
  &\ben_{\rm{min}}=\frac{\mu\;\log_e 2}{\f_M(\varepsilon)},\\ \label{eq:min_S0}
   &S_0=\frac{2\mu\log_e 2}{\sigma^2 \theta+\mu^2 \log_e 2},
\end{align}
where $\mu$ and $\sigma^2$ are given by (\ref{eq:mu_4}) and (\ref{eq:sgm_2}), respectively.
\end{Lem}

\begin{proof}
See Appendix \ref{APD_1}.
\end{proof}

We immediately notice that the minimum energy per bit $\ben_{\rm{min}}$ does not depend on the QoS exponent $\theta$, and hence is not affected by the presence of QoS constraints. On the other hand, via $\mu$ and $\f_M(\varepsilon)$, $\ben_{\rm{min}}$ is a function of the deadline constraint $M$ and the outage limit $\epsilon$. This dependence will be explored in the numerical results. We further notice that the wideband slope $S_0$ diminishes with increasing $\theta$. Hence, stricter QoS constraints lead to smaller slopes, increasing the energy per bit requirements at the same throughput level.

\subsection{Energy Efficiency of HARQ-CC with ON-OFF Discrete-Time Markov Source}\label{SC_DMS}
When the arrival rate $a_i$ is not constant, the computation of the throughput is more involved. Generally, we need to express the LMGFs of the random arrival processes and random departure processes (or equivalently random wireless transmissions), and then solve (\ref{eq:qos}) in order to determine the maximum average arrival rate $r_{\avg}$ that can be supported by the wireless transmissions under statistical queuing constraints. In these cases, derivation of the minimum bit energy and wideband slope only involves the first and second order derivatives of $r_{\text{avg}}$ evaluated at $\tsnr=0$, which can be obtained easily by taking the derivatives of both sides of (\ref{eq:qos}) and letting $\tsnr\rightarrow 0$. In this subsection, we analyze the energy efficiency of HARQ-CC with fixed outage probability when we have ON-OFF discrete-time Markov sources.

In this case, the Markov source only has two states, namely, ON and OFF states. We define state $1$ as the OFF state, in which the source keeps silent. When the source is in ON state, or equivalently state $2$, the arrival rate is $a_i=r\;\text{(bit/s/Hz)}$. The state transition probability matrix of this Markov source can be written as
\begin{align}
\mathbf{G}=\left(
             \begin{array}{cc}
               p_{11} & p_{12} \\
               p_{21} & p_{22} \\
             \end{array}
           \right),
\end{align}
where $p_{11}$ and $p_{22}$ denote the probabilities that the source remains in the same state (OFF and ON states, respectively) in the next time block, and $p_{12}$ and $p_{21}$ are the probabilities that source will transition to a different state in the next time block. Using the properties of Markov processes, we can express the probability of the ON state as
\begin{equation}\label{eq:p_on}
P_{ON}=\frac{1-p_{11}}{2-p_{11}-p_{22}}.
\end{equation}
Then, the average arrival rate of this ON-OFF Markov source is
\begin{equation}
r_{\text{avg}}=r P_{ON}=r \frac{1-p_{11}}{2-p_{11}-p_{22}}.
\end{equation}

Since the departure and arrival processes at the transmitter are independent, the expressions of $\mu$ and $\sigma^2$ in (\ref{eq:mu_4}) and (\ref{eq:sgm_2}) are still valid for this case.
\begin{Lem}\label{theo:discrete2}
For small QoS exponent $\theta$ and ON-OFF discrete-time Markov source, the minimum energy per bit and wideband slope of the HARQ-CC scheme with the outage constraint $\epsilon$ are given, respectively, by
\begin{align}
  &\ben_{\rm{min}}=\frac{\mu\;\log_e 2}{\f_M(\varepsilon)},\\
   &S_0=\frac{2\log_e 2}{\frac{\sigma^2 \theta+\mu^2 \log_e 2}{\mu}+\theta\zeta},
\end{align}
where $\mu$ and $\sigma^2$ are given by (\ref{eq:mu_4}) and (\ref{eq:sgm_2}), respectively, and $\zeta$ is defined as
\begin{align}\label{eta}
\zeta=\frac{(1-p_{22})(p_{11}+p_{22})}{(1-p_{11})(2-p_{11}-p_{22})}.
\end{align}
\end{Lem}
\begin{proof}
See Appendix \ref{APD_2}.
\end{proof}

From Theorem \ref{theo:discrete2}, we observe that source randomness does not have any influence on the minimum energy per bit. The minimum energy per bit shown in Theorem \ref{theo:discrete2} is the same as in the case of constant-rate arrivals. Source burstiness has influence only on the wideband slope. Compared with the constant arrival case, there is an additional term $\theta\zeta$ in the denominator. Since both of $p_{11}$ and $p_{22}$ are between $0$ and $1$, it is easy to verify that $\theta\zeta\geq 0$, which means that source burstiness always degrades the wideband slope and makes the system less energy-efficient. When $P_{ON}=1$, we have $\zeta=0$, which corresponds to the constant arrival case, and the results in Theorem \ref{theo:discrete2} specialize to those in the case of the constant-rate arrivals.

\subsection{Energy Efficiency of HARQ-CC with ON-OFF Fluid Markov Source}\label{SC_FMS}
In this section, we consider the ON-OFF fluid Markov sources. Different from the discrete-time Markov source whose state does not change in a given time block and state transitions occur in discrete time steps, fluid Markov source may stay in a state over a continuous duration of time. In other words, the source can change its state at any time. Here, the definitions of ON and OFF states are the same as for the ON-OFF discrete-time source. The generating matrix of this continuous-time Markov process is given by
\begin{align}\label{eq:GMatrix}
\mathbf{G}=\left(
             \begin{array}{cc}
               -\alpha & \alpha \\
               \beta & -\beta \\
             \end{array}
           \right),
\end{align}
and the ON state probability is $P_{ON}=\frac{\alpha}{\alpha+\beta}$. In this case, the average arrival rate is
\begin{align}
r_{\text{avg}}=&\:r P_{ON}\notag\\
       =&\:r\frac{\alpha}{\alpha+\beta}.
\end{align}
Using a similar approach as for the discrete-time Markov source, we can find the minimum energy per bit and wideband slope for the ON-OFF fluid Markov source as in the following result.

\begin{Lem}\label{theo:discrete4}
For small QoS exponent $\theta$ and ON-OFF fluid Markov source, the minimum energy per bit and wideband slope of the HARQ-CC scheme with the outage constraint $\epsilon$ are given, respectively, by
\begin{align}
  &\ben_{\rm{min}}=\frac{\mu\;\log_e 2}{\f_M(\varepsilon)},\\
   &S_0=\frac{2\log_e 2}{\frac{\sigma^2 \theta+\mu^2 \log_e 2}{\mu}+\frac{2\theta\beta}{\alpha(\alpha+\beta)}},\label{eq:s0_MF}
\end{align}
where $\mu$ and $\sigma^2$ are given by (\ref{eq:mu_4}) and (\ref{eq:sgm_2}), respectively.
\end{Lem}
\begin{proof}
See Appendix \ref{APD_3}.
\end{proof}

Similar to the ON-OFF discrete-time Markov source, we notice that source burstiness does not change the minimum energy per bit, and it only results in the addition of the positive term $\frac{2\theta\beta}{\alpha(\alpha+\beta)}$ in the denominator of the wideband slope expression in (\ref{eq:s0_MF}). When $P_{ON}=1$, arrival rates become constant, and this additional term vanishes. Therefore, source burstiness has again a negative influence on the energy efficiency.

\subsection{Energy Efficiency of HARQ-CC with ON-OFF Markov Modulated Poisson Sources (MMPS)}
In this subsection, we investigate the energy efficiency of ON-OFF MMPS models whose arrival rates are described as a Poisson process with intensity $\nu$ in the ON state while there is no arrival in the OFF state. State transitions are governed by a continuous-time Markov chain as in the Markov fluid model. However, compared to the ON-OFF Markov fluid source analyzed in Section \ref{SC_FMS}, MMPS can be seen to have a higher degree of burstiness since its arrival rate, rather than being a constant, is random in the ON state. Here, the expressions of the generating matrix and ON state probability are the same as in Section \ref{SC_FMS}. In this case, the average arrival rate is
\begin{align}\label{eq:ravg_MMPS}
  r_{\text{avg}}=&\:\nu P_{ON}\notag\\
       =&\:\nu\frac{\alpha}{\alpha+\beta},
\end{align}
where $\nu$ is the Poisson intensity in the ON state. The following result identifies the the minimum energy per bit and wideband slope for the ON-OFF MMPS models.

\begin{Lem}\label{theo:MMPS1}
For small QoS exponent $\theta$ and ON-OFF MMPS, the minimum energy per bit and wideband slope of the HARQ-CC scheme with the outage constraint $\epsilon$ are given, respectively, by
\begin{align}
  &\ben_{\rm{min}}=\frac{e^\theta-1}{\theta}\:\frac{\mu\;\log_e 2}{\f_M(\varepsilon)},\\
  &S_0=\frac{\theta}{e^\theta-1}\:\frac{2\log_e 2}{\frac{\sigma^2 \theta+\mu^2 \log_e 2}{\mu}+\frac{2\theta\beta}{\alpha(\alpha+\beta)}},
\end{align}
where $\mu$ and $\sigma^2$ are given by (\ref{eq:mu_4}) and (\ref{eq:sgm_2}), respectively.
\end{Lem}
\begin{proof}
See Appendix \ref{APD_4}.
\end{proof}

Comparing the results of Theorems \ref{theo:discrete4} and \ref{theo:MMPS1}, we notice that Poisson arrival model leads to the introduction of the additional factor of $\frac{\theta}{e^\theta-1}$ in the expressions of the minimum energy per bit and wideband slope. For $\theta\geq0$, we have $\frac{\theta}{e^\theta-1}\leq1$, resulting in a larger minimum energy per bit and smaller wideband slope for the ON-OFF MMPS compared to those for the ON-OFF Markov fluid source. Since the factor $\frac{\theta}{e^\theta-1}$ is a decreasing function of $\theta$, the performance gap grows further as the queuing constraint gets stricter. Moreover, as a stark contrast to the observations in Sections \ref{SC_DMS} and \ref{SC_FMS}, the minimum energy per bit depends on $\theta$ when MMPS arrival model is considered.

\section{Energy Efficiency of HARQ-CC Scheme with Optimal Transmission Rate}\label{sec:energy-eff-optimal-rate}
In this section, instead of maintaining a fixed outage probability, we study the energy efficiency of the HARQ-CC scheme with the optimal transmission rate, which maximizes the effective capacity. Because the fixed rate only has influence on the LMGF of the departure process, or equivalently the effective capacity, it is very easy to verify that the optimal transmission rate that maximizes the effective capacity also maximizes the average arrival rate and hence the system throughput. It can be easily seen that as $R \rightarrow0$, effective capacity $C_E(\tSNR) \rightarrow 0$. Moreover, as $R\rightarrow \infty$, then we again have $C_E(\tSNR) \rightarrow 0$ since transmission failures after $M$ HARQ rounds and hence outage events occur more and more frequently with increasing $R$, lowering the throughput.
Therefore, there exists a finite optimal rate, $\R(\tsnr)$, which maximizes the effective capacity. Assume that $\R(\tsnr)$ has the following first-order expansion at $\tSNR=0$
\begin{align}
\R(\tsnr)=a\:\tsnr+o(\tsnr)
\end{align}
where $a$ is the value of the first derivative of $\R(\tsnr)$ with respect to $\tSNR$ at $\tSNR = 0$. Given this optimal transmission rate, the outage probability can be expressed as
\begin{align}\label{e2}
  \varepsilon(\tsnr)=&\Pr\left\{\log_{2}\left(1+\tsnr\sum_{i=1}^{M}z_i\right)<\R(\tsnr) \right\}\notag\\
                    =&F_M\left(\frac{2^{\R(\tsnr)}-1}{\tsnr} \right),
\end{align}
whose limit as $\tSNR$ vanishes is
\begin{equation}
  \lim_{\tsnr\rightarrow 0}\varepsilon(\tsnr)=F_M(a \log_e2).
\end{equation}
The expression in (\ref{e2}) shows that the outage probability is a monotonic increasing function of the transmission rate for fixed $\tSNR$. Hence, we can see that searching for the optimal rate for a certain $\tSNR$ is equivalent to searching for the optimal outage probability.

We initially start with the constant-rate arrival model. In this case, the throughput $r_{\text{avg}}$ is equal to the effective capacity $C_E$. Now, given the optimal rate, we can characterize the effective capacity for small $\theta$ as
\begin{equation}\label{eq:C_E_opt_rate}
  C_E(\tsnr)=\frac{\R(\tsnr)}{\mu}-\frac{(\R(\tsnr))^2\sigma^2\theta}{2\mu^3}.
\end{equation}
In order to find $\mu$ and $\sigma^2$, we first derive the probability $\Pr\{\hat{T}=k\,M+t\}$ by following similar steps as in Section \ref{sec:EE_fixed_outage}:
\begin{align}
\small
\begin{split}
&\Pr\{\hat{T}=kM+t\}= (\Pr\{T>M\})^k \Pr\{T=t\} \\
&= \varepsilon ^k \Pr\{T=t\} \\
&= \varepsilon ^k \Bigg\{\!\! \Pr\left(\sum_{i=1}^{t}z_i\geqslant\frac{2^{\R(\ssnr)}\!-\!1}{\ssnr} \!\right)\!-\!\Pr\left(\sum_{i=1}^{t-1}z_i\geqslant\frac{2^{\R(\ssnr)}\!-\!1}{\ssnr} \!\right) \!\!\Bigg\} \\
&= \varepsilon ^k \left(\text{F}_{t-1}\left(\f_M(\varepsilon)\right)-\text{F}_{t}\left(\f_M(\varepsilon)\right)\right)
\end{split}
\normalsize
\end{align}
where $\varepsilon$ is given by (\ref{e2}).
\begin{Rem}
If we further assume Rayleigh fading channel, then the probability mass function (pmf) of $\hat{T}$ can be simplified according to the relationship between Poisson distribution and chi-square distribution \cite{poisson}. This is similar to Remark \ref{rem1}, and the simplified result is given by
\begin{align}
\Pr\{\hat{T}=kM+t\}= \varepsilon ^k \frac{\tilde{\lambda}^{t-1}}{(t-1)!}\rme^{-\tilde{\lambda}},\label{eq:pmfT}
\end{align}
where $\tilde{\lambda}(\tsnr)=\frac{2^{\R(\ssnr)}-1}{\E\{z\}\ssnr}$ and $\lim_{\tsnr\rightarrow 0}\tilde{\lambda}(\tsnr)=\frac{a}{\E\{z\}}\log_e 2$.
\end{Rem}

By plugging the pmf in (\ref{eq:pmfT}) into (\ref{eq:mu_4}) and (\ref{eq:sgm_2}), $\mu$ and $\sigma^2$ can be found, respectively. Different from the analysis in Section \ref{sec:EE_fixed_outage}, $\mu$ and $\sigma^2$ now depend on $\tSNR$. In the following, we provide a characterization of the minimum energy per bit for the constant-rate arrival model.

\begin{Lem} \label{theo:discrete3}
For small QoS exponent $\theta$ and constant-rate arrivals, the minimum energy per bit of the HARQ-CC scheme with optimal transmission rate is given by
\begin{align} \label{eq:EbN0_min_opt_rate}
  \ben_{\rm{min}}=&\frac{\mu(0)}{a}
\end{align}
where $\mu(0)$ is the value of $\mu (\tSNR)$ evaluated at zero $\tsnr$, $\mu (\tSNR)$ is given in (\ref{eq:mu_4}), and $a$ is the first derivative of $\R(\tsnr)$ with respect to $\tsnr$ at $\tsnr = 0$.
\end{Lem}
\begin{proof}
See Appendix \ref{APD_5}.
\end{proof}

For the case of the optimal rate scheme, we note that most of the analysis remains in the same form as that of the fixed outage probability case, and the only difference is that the outage probability is a function of $\tsnr$. Next, we consider Markov source models.

\begin{Lem}\label{theo:discrete5}
For small QoS exponent $\theta$, the minimum energy per bit of the HARQ-CC scheme with optimal transmission rate is given by
\begin{align} \label{eq:EbN0_min_opt_rate_Markov}
  \ben_{\rm{min}}=&\frac{\mu(0)}{a}
\end{align}
for both discrete-time and fluid ON-OFF Markov sources, and is given by
\begin{align} \label{eq:EbN0_min_opt_MMPS}
  \ben_{\rm{min}}=&\frac{e^\theta-1}{\theta}\frac{\mu(0)}{a}
\end{align}
for the ON-OFF MMPS.

In the above expressions, $\mu(0)$ is the value of $\mu (\tSNR)$ evaluated at zero $\tsnr$, $\mu (\tSNR)$ is given in (\ref{eq:mu_4}), and $a$ is the first derivative of $\R(\tsnr)$ with respect to $\tsnr$ at $\tsnr = 0$.
\end{Lem}
\begin{proof}
See Appendix \ref{APD_6}.
\end{proof}


From Theorems \ref{theo:discrete3} and \ref{theo:discrete5}, we conclude that source burstiness does not affect the minimum energy per bit for both ON-OFF discrete Markov and ON-OFF Markov fluid sources. On the other hand, for the ON-OFF MMPS, the Poisson property leads to the presence of the factor $\frac{e^\theta-1}{\theta}$ in the expression of the minimum energy per bit, indicating lower energy efficiency compared to that of ON-OFF Markov fluid sources. This is a similar observation as in the fixed outage probability case and can again be attributed to the more bursty nature of MMPS.

\section{Numerical Results}\label{sec:numerical}
In this section, we present numerical results to illustrate the energy efficiency of HARQ-CC in the presence of QoS constraints. In the first subsection, numerical results for the constant-rate arrival model are provided to demonstrate the influence of the deadline constraint $M$ and outage probability $\varepsilon$. In the second subsection, we concentrate on the impact of random arrivals and source burstiness. Within this section, unless mentioned explicitly, QoS exponent is set to $\theta = 0.1$.
\subsection{Constant Arrival Models}
In this subsection, we analyze the energy efficiency of HARQ-CC scheme with fixed transmission rate and constant arrival rate. Making use of the characterizations in \cite{HARQ}, we also numerically evaluate the performance of HARQ-IR. In the simulations, we consider Rayleigh fading channel with exponentially distributed fading power having a mean value of $\E\{z\}=1$.

In Fig. \ref{fig1}, we plot the maximum average arrival rate $r_{\text{avg}}$ (or equivalently throughput) as a function of the energy per bit $\ben$ for HARQ-CC and HARQ-IR schemes under two different outage constraints $\epsilon$ and deadline constraints $M$. Since the expected value $\mu$ and variance $\sigma^2$ of the random transmission time are not available in closed-form for HARQ-IR, throughput for this case is evaluated numerically. Analytical throughput curves for HARQ-CC are also validated via Monte Carlo simulations with $20 \times 10^6$ samples. We notice that analytical and simulation results agree perfectly. In the figure, it is seen that HARQ-CC and HARQ-IR schemes approach the same minimum energy per bit under the same outage and deadline constraints. An intuitive explanation of this observation is that for vanishingly small $x$, we have $\log_2(1+x) \sim x\log_2 e$. Hence, for low $\tsnr$ values, we have $\sum_i\log_2\left(1+\tsnr\;z_i\right)\sim \tsnr\sum_i\;z_i\log_2 e$ and  $\log_2\left(1+\tsnr\sum_i z_i\right) \sim \tsnr\sum_i z_i\log_2 e$, indicating that these two HARQ schemes are expected to have similar performances at vanishingly small $\tsnr$ values.  We also observe that HARQ-IR has a higher wideband slope. Hence, at low but nonzero values of $\tsnr$, HARQ-IR provides better energy efficiency compared to HARQ-CC.

\begin{figure}
\center
\includegraphics[width=\figsize\textwidth]{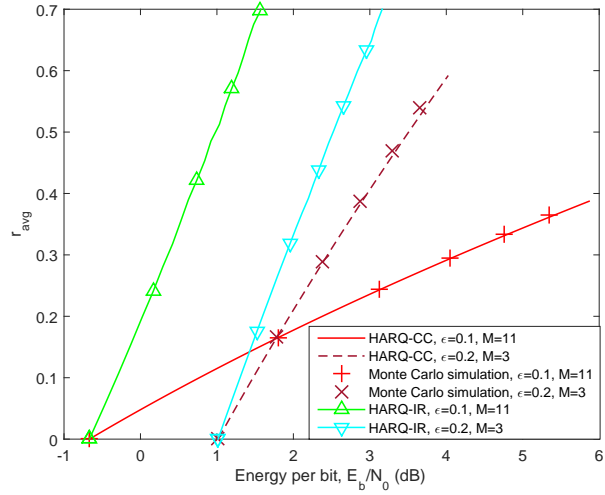}
\caption{Maximum average arrival rate $r_{\text{avg}}$ vs. energy per bit $\ben$}\label{fig1}
\end{figure}

In Fig. \ref{fig2}, we display the minimum energy per bit $\ben_{\rm{min}}$ and wideband slope $S_0$ of HARQ-CC as a function of the outage probability constraint $\epsilon$ for three different values of the deadline constraint $M$. Recall that minimum energy per bit does not depend on the QoS constraints while wideband slope does. We consider two different QoS exponents for the wideband slope. For higher values of $\theta$ under the same deadline constraint $M$, we have smaller wideband slopes as expected since higher values of $\theta$ indicate stricter QoS constraints. It is observed from the figure that the minimum energy per bit first decreases with increasing $\epsilon$ and then starts increasing after a certain threshold point. On the other hand, wideband slope always decreases with increasing $\epsilon$.
\begin{figure}
\begin{centering}
\center\includegraphics[width=\figsize\textwidth]{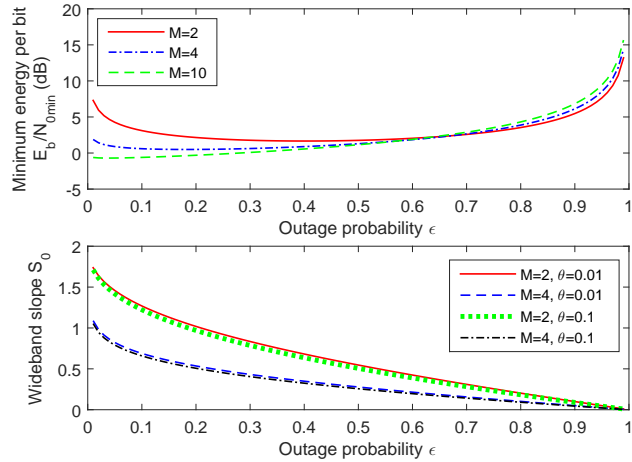}
\caption{Minimum energy per bit $\ben_{\rm{min}}$ and wideband slope $S_0$ vs. outage probability $\epsilon$}\label{fig2}
\end{centering}
\end{figure}
\begin{figure}
\begin{centering}
\center
\includegraphics[width=\figsize\textwidth]{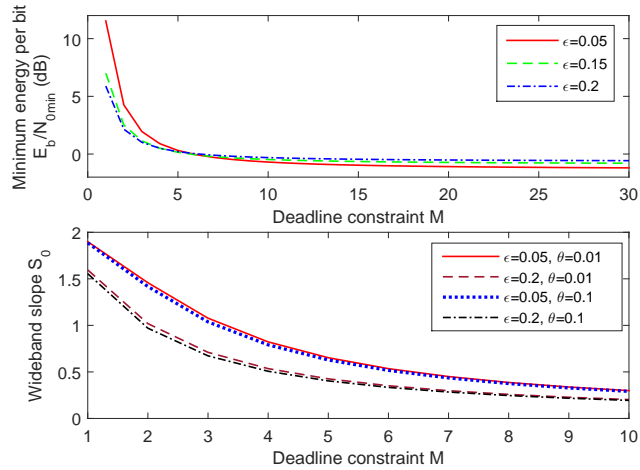}
\caption{Minimum energy per bit $\ben_{\rm{min}}$ and wideband slope $S_0$ vs. deadline constraint $M$}\label{fig3}
\end{centering}
\end{figure}
In Fig. \ref{fig3}, the minimum energy per bit and wideband slope are plotted as a function of the deadline constraint $M$ for the HARQ-CC scheme. It is seen that both the minimum energy per bit and wideband slope decrease with increasing $M$. Hence, by reducing the minimum energy per bit, relaxed deadline constraints lead to improvements in energy efficiency in the vicinity of $\ben_{\rm{min}}$.


In Fig. \ref{fig6},  we display the maximum average arrival rate $r_{\text{avg}}$ as a function of the energy per bit $\ben$ for HARQ-CC with fixed outage probability and also HARQ-CC with optimal transmission rate. In the case of fixed outage probability, the outage probability $\varepsilon$ is chosen such that $\ben_{\rm{min}}$ is smallest. In the other case, the optimal transmission rate which maximizes the effective capacity is chosen. As noted before, this is actually equivalent to optimizing the outage probability. Therefore, as expected, when $\tsnr\rightarrow 0$, both HARQ-CC with the optimal transmission rate and HARQ-CC with fixed outage probability achieve the same minimum energy per bit. However, HARQ-CC with the optimal transmission rate has a higher wideband slope. Therefore, when $\tsnr$ is small but nonzero, HARQ-CC with the optimal transmission rate outperforms and provides better energy efficiency.

\begin{figure}
\center
\includegraphics[width=\figsize\textwidth]{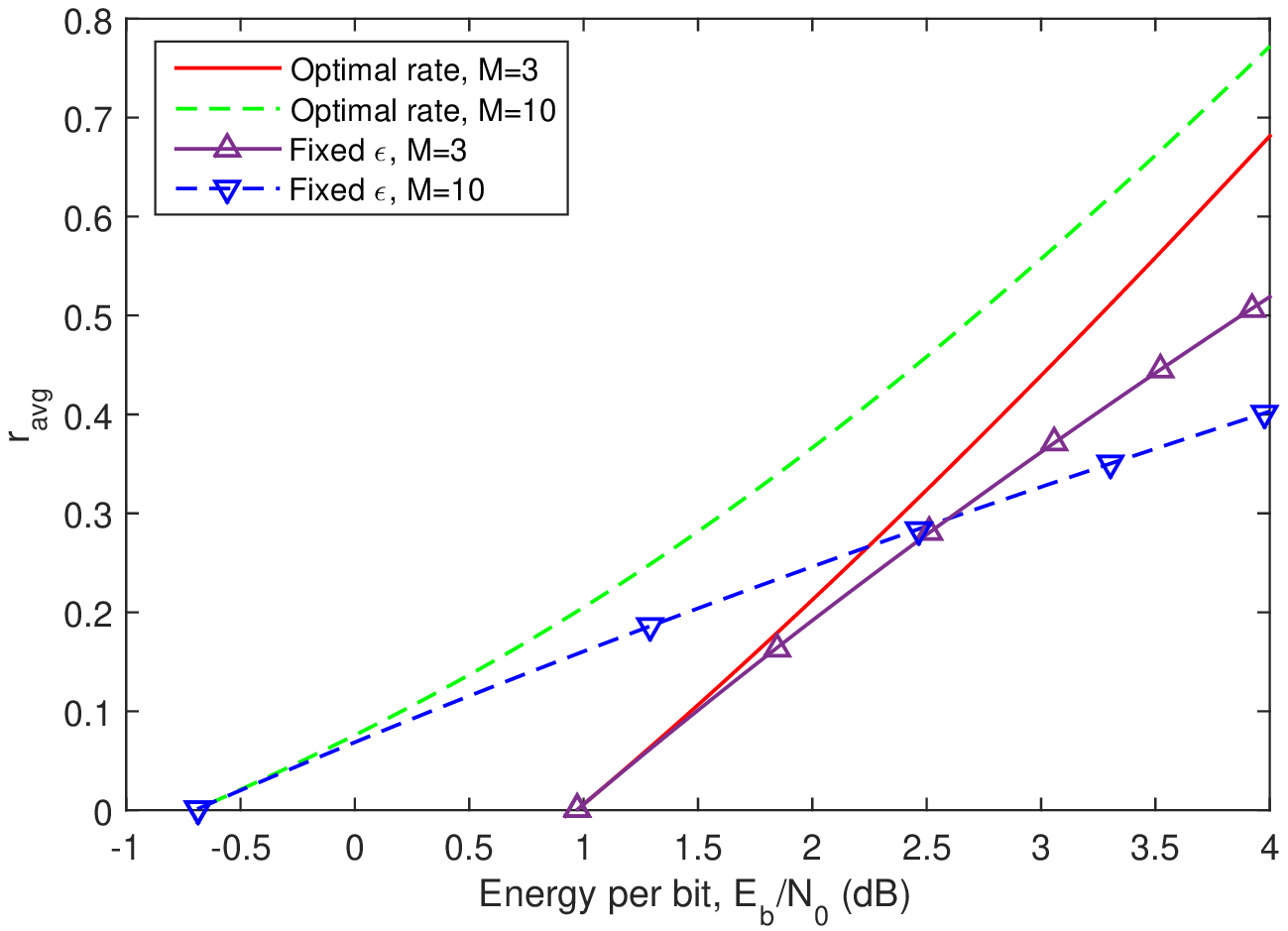}
\caption{Maximum average arrival rate $r_{\text{avg}}$ vs. energy per bit $\ben$}\label{fig6}
\end{figure}

\subsection{Random Arrival Models}
\begin{figure}
\center
\includegraphics[width=\figsize\textwidth]{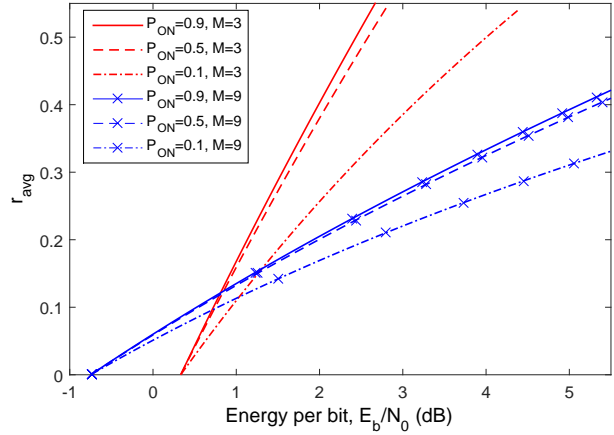}
\caption{Maximum average arrival rate $r_{\text{avg}}$ vs. energy per bit $\ben$ for ON-OFF discrete-time Markov source with fixed outage probability $\varepsilon=0.1$}\label{fig7}
\end{figure}
\begin{figure}
\center
\includegraphics[width=\figsize\textwidth]{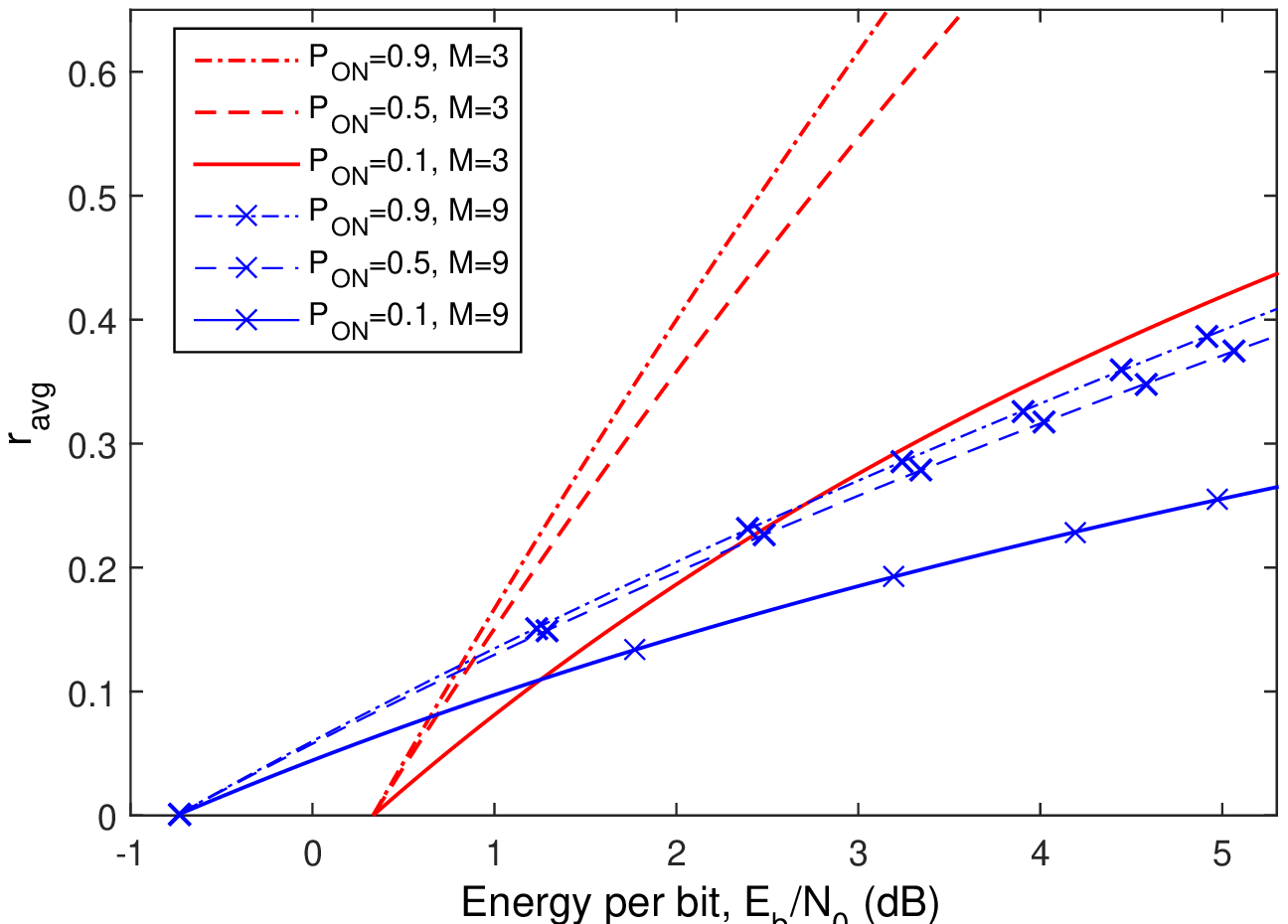}
\caption{Maximum average arrival rate $r_{\text{avg}}$ vs. energy per bit $\ben$ for ON-OFF Markov fluid source with fixed outage probability $\varepsilon=0.1$}\label{fig8}
\end{figure}
\begin{figure}
\center
\includegraphics[width=\figsize\textwidth]{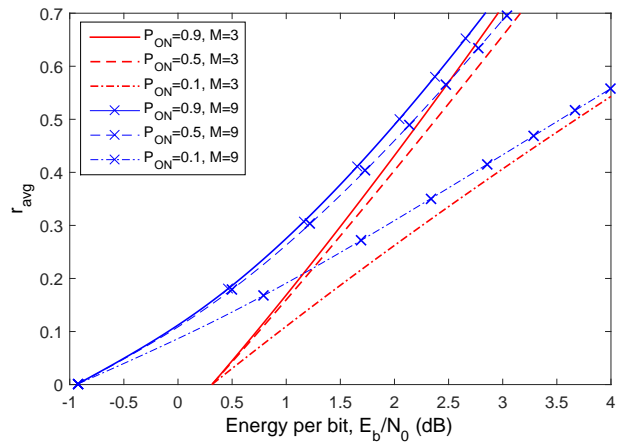}
\caption{Maximum average arrival rate $r_{\text{avg}}$ vs. energy per bit $\ben$ for ON-OFF discrete-time Markov source with optimal transmission rate}\label{fig9}
\end{figure}
\begin{figure}
\center
\includegraphics[width=\figsize\textwidth]{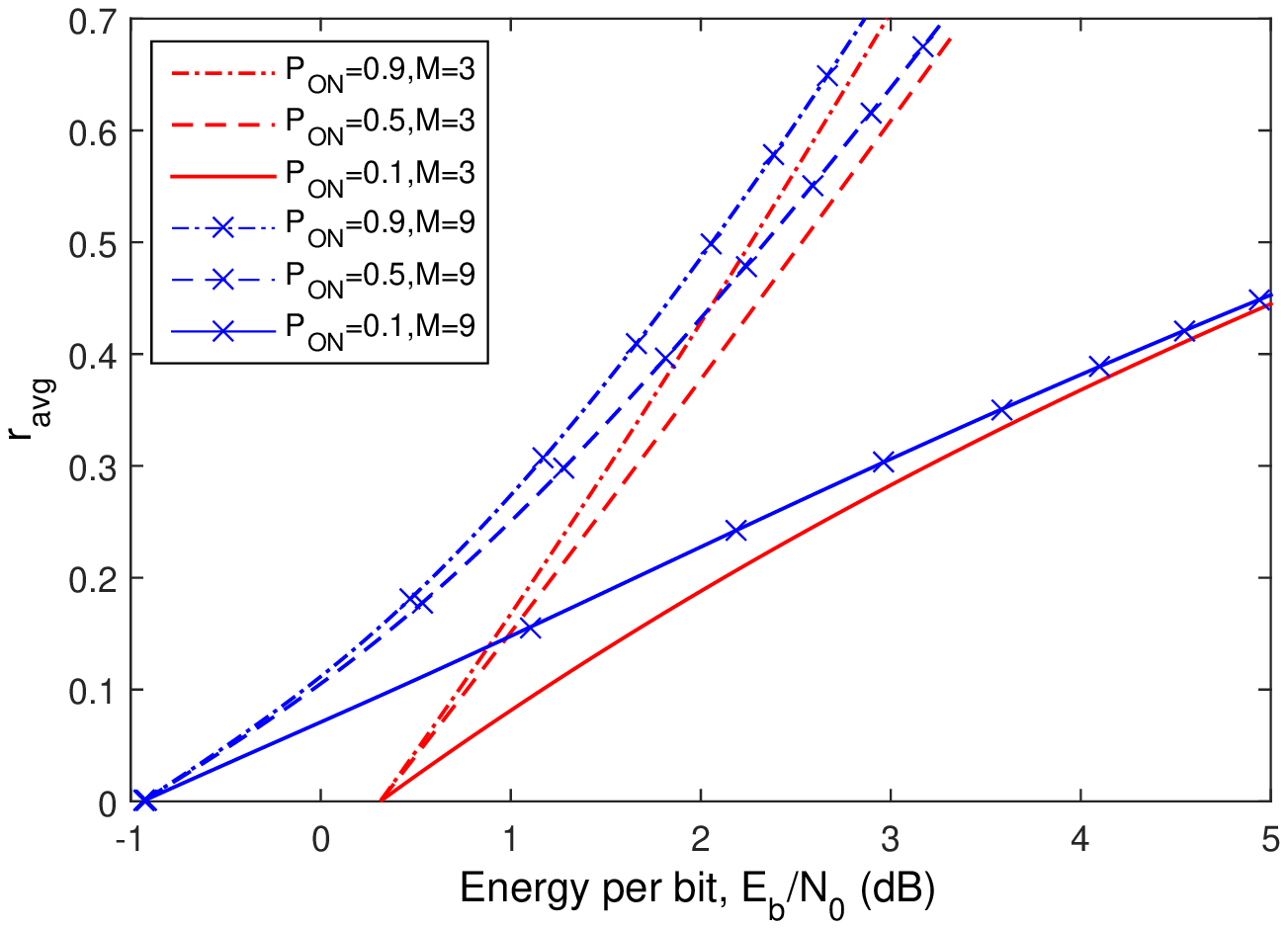}
\caption{Maximum average arrival rate $r_{\text{avg}}$ vs. energy per bit $\ben$ for ON-OFF Markov fluid source with optimal transmission rate}\label{fig10}
\end{figure}

In this section, we investigate the impact of source randomness/burstiness on the energy efficiency. Within this subsection, we assume a Nakagami-$m$ fading channel with $m=2$. Also, we assume $\E\{z\}=1$. For all fixed outage probability results, we fix $\varepsilon=0.1$.

Fig. \ref{fig7}--Fig. \ref{fig10} demonstrate the influence of source burstiness considering both ON-OFF discrete-time Markov and Markov fluid sources. Here we use the ON state probability $P_{ON}$ as a measure of source burstiness. We set $p_{11}+p_{22}=1$ in the discrete-time Markov source model, and $\alpha+\beta=1$ for the Markov fluid source. Under these assumptions, the Markov source with higher $P_{ON}$ has less burstiness.

First, we observe that minimum energy per bit does not depend on source burstiness, which has been proved analytically in previous sections. Also, we note that source burstiness makes the wideband slope smaller in both discrete-time and fluid models. With the same average arrival rate, the source with smaller $P_{ON}$ can have much larger arrival rate in the ON state, which makes it more difficult to satisfy the queuing constraint.

From Theorems \ref{theo:discrete2} and \ref{theo:discrete4}, we have seen that the impact of source burstiness and channel conditions are essentially separated. For both discrete-time and fluid models, the denominator of the wideband slope expressions in these theorems can be divided into two components. One can be called as the arrival component, which only depends on the arrival process, and the other can be called as the departure component, which only depends on the departure process at the buffer. In the discrete-time model, the arrival component is $\theta\zeta$, and the departure component is $\frac{\sigma^2 \theta+\mu^2 \log_e 2}{\mu}$; for Markov fluid model, the arrival component is $\frac{2\theta\beta}{\alpha(\alpha+\beta)}$, and the departure component is $\frac{\sigma^2 \theta+\mu^2 \log_e 2}{\mu}$. When $P_{ON}=1$, our random arrival model specializes to the constant-rate arrival model, and the formulas become the same as in the constant-rate arrival case, because the arrival components are equal to $0$. This observation implies that our analysis on the influence of deadline constraints and outage probability for the constant-rate arrival case is applicable to the random arrival model. From Fig. \ref{fig7} through Fig. \ref{fig10}, we notice that both the minimum energy per bit and wideband slope decrease as the deadline constraint $M$ increases, which agrees with our analysis in the previous subsection.

Although the impact of source burstiness and channel conditions are separate, we observe in Figs. \ref{fig7} and \ref{fig8} that the influence of source burstiness becomes smaller when $M$ increases. Since both of the arrival and departure components are in the denominator, when one component increases, the other one becomes less significant. When $M$ increases, the departure component becomes larger for both discrete-time and fluid models. Therefore, we note that the effect of source burstiness becomes deemphasized as $M$ increases from $3$ to $9$.

In Figs. \ref{fig9} and \ref{fig10}, we observe that for different $M$ values, the throughput curves with the same $P_{ON}$ are very close to each other. This implies that for sufficiently high SNR, the effect of the deadline constraint becomes very small. Since the outage probability approaches $0$, when the transmitter has enough energy to complete the transmission within a short time, deadline constraint becomes insignificant.
\begin{figure}
\center
\includegraphics[width=\figsize\textwidth]{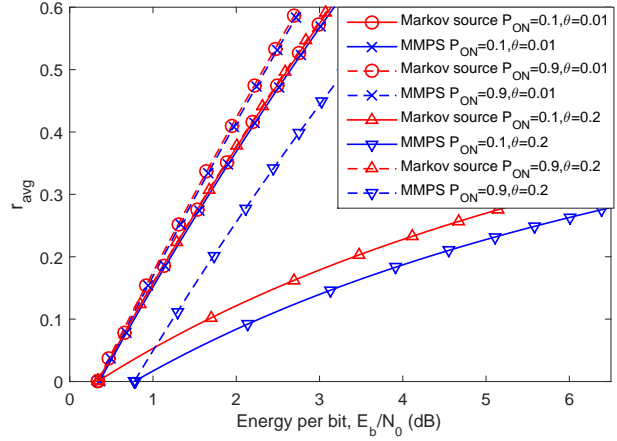}
\caption{Maximum average arrival rate $r_{\text{avg}}$ vs. energy per bit $\ben$ for ON-OFF Markov fluid source and MMPS with fixed outage probability $\varepsilon=0.1$}\label{fig11}
\end{figure}

\begin{figure}
\center
\includegraphics[width=\figsize\textwidth]{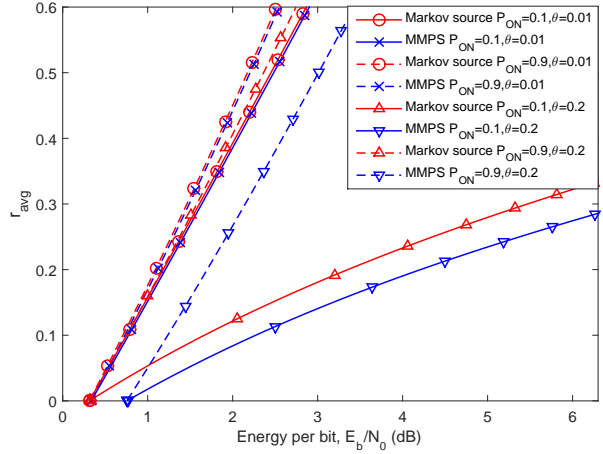}
\caption{Maximum average arrival rate $r_{\text{avg}}$ vs. energy per bit $\ben$ for ON-OFF Markov fluid source and MMPS with optimal transmission rate}\label{fig12}
\end{figure}

Finally, in Figs. \ref{fig11} and \ref{fig12}, we compare the performances of ON-OFF Markov fluid source and MMPS for both cases of fixed outage probability and optimal transmission rates. As mentioned in previous sections, compared to the the minimum energy per bit and wideband slope of the ON-OFF Markov fluid source, the corresponding results for MMPS are scaled by the factor $\frac{e^\theta-1}{\theta}$ and its reciprocal, respectively. When $\theta$ is close to $0$, both $\frac{e^\theta-1}{\theta}$ and its reciprocal approach $1$. For this reason, the throughput curves of ON-OFF Markov fluid source and MMPS stay very close to each other in both figures when $\theta=0.01$. As $\theta$ increases, the factor $\frac{e^\theta-1}{\theta}$ grows, which leads to larger gap between the throughput curves of these two types of Markov sources. For instance, we can easily observe from Figs. \ref{fig11} and \ref{fig12} that there is a $0.44\;\text{dB}$ difference between the corresponding minimum energy per bit values when $\theta=0.2$.

\section{Conclusion}\label{sec:conc}
In this paper, we have analyzed the energy efficiency of the HARQ-CC scheme under outage, deadline, and statistical queuing constraints in the low-power and low-$\theta$ regimes by employing the notions of effective capacity and effective bandwidth from the stochastic network calculus while considering both constant-rate and random data arrivals to the buffer. First, we have determined the minimum energy per bit and wideband slope achieved with HARQ-CC for fixed outage probability and both constant-rate and Markov source models. From the results, we have shown that source burstiness does not affect the minimum energy per bit when ON-OFF discrete time and  Markov fluid sources are considered. On the other hand, due to the Poisson arrivals and the resulting higher level of burstiness, MMPS is shown to have worse energy efficiency compared to the ON-OFF Markov fluid source. Moreover, among the considered arrival models, MMPS is the only source for which the minimum energy per bit depends on the QoS exponent $\theta$ and grows with stricter QoS constraints. In contrast to the characterizations regarding the minimum energy per bit, we have shown that wideband slope in all cases varies with the QoS exponent $\theta$ and source statistics. For instance, stricter queuing constraints (i.e., larger values of $\theta$) and increased source burstiness tend to lower the wideband slope, incurring loss in the energy efficiency. The impact of source burstiness is clearly identified with additional terms introduced in the denominators of the wideband slope expressions.

For the case in which the transmission rate is optimized, we have determined the minimum energy per bit in terms of the mean transmission time and transmission rate characteristics at vanishingly small $\tsnr$ levels. Constant-rate, discrete-time Markov, and Markov fluid arrival models all lead to the same minimum energy per bit while MMPS results in a larger $\ben_{\min}$ value.

Through numerical results, we have illustrated that while HARQ-IR and HARQ-CC achieve the same minimum energy per bit, HARQ-IR outperforms HARQ-CC at low but nonzero $\tsnr$ levels due to having a higher wideband slope. In the numerical analysis, we have also investigated the interactions between deadline constraints, target outage probability, QoS constraints, source burstiness, and energy efficiency.

\begin{appendices}
\section{Proof of Theorem \ref{theo:discrete}}\label{APD_1}
\begin{proof}
In order to derive the minimum energy per bit and wideband slope expressions, we need to obtain the first and second derivatives of $r_{\text{avg}}(\tSNR)$ with respect to $\tSNR$ at zero $\tSNR$. For the constant-rate arrival model, $r_{\text{avg}}$ is given by (\ref{eq:CE_HARQ}). In this regard, the first and second derivatives of $r_{\text{avg}}(\tSNR)$ with respect to $\tSNR$, are given, respectively, by
\begin{align}
  \dot{r}_{\text{avg}}(\tSNR)=&\frac{\f_M(\varepsilon)}{(1+\f_M(\varepsilon)\tsnr)\mu\log_e 2}\notag\\
                 &-\frac{\f_M(\varepsilon)\theta\sigma^2\left(\log_e (1+\f_M(\varepsilon)\tsnr)\right)^2}{(1+\f_M(\varepsilon)\tsnr)\mu^3(\log_e 2)^2},
\end{align}
\begin{align}
  \ddot{r}_{\text{avg}}(\tSNR)=&\frac{\left(\f_M(\varepsilon)\right)^2\theta\sigma^2}{(1+\f_M(\varepsilon)\tsnr)^2\mu^3(\log_e 2)^2}\notag\\
                  &-\frac{\left(\f_M(\varepsilon)\right)^2}{(1+\f_M(\varepsilon)\tsnr)^2\mu\log_e 2}\notag\\
                  &+\frac{\left(\f_M(\varepsilon)\right)^2\theta\sigma^2\log_e(1+\f_M(\varepsilon)\tsnr)}{(1+\f_M(\varepsilon)\tsnr)^2\mu^3(\log_e 2)^2}.
\end{align}
Then, taking the limit as $\tsnr\rightarrow 0$ results in the following expressions:
\begin{equation} \label{eq:firstder_Ec_0}
  \dot{r}_{\text{avg}}(0)=\frac{\f_M(\varepsilon)}{\mu\log_e 2},
\end{equation}
and
\begin{equation} \label{eq:secondder_Ec_0}
  \ddot{r}_{\text{avg}}(0)=-\frac{\left(\f_M(\varepsilon)\right)^2(\theta\sigma^2+\mu^2\log_e 2)}{\mu^3(\log_e 2)^2}.
\end{equation}
Inserting the expressions in (\ref{eq:firstder_Ec_0}) and (\ref{eq:secondder_Ec_0}) into (\ref{bitenergy}) and (\ref{widebandslop}), the minimum bit energy in (\ref{eq:min_Eb}) and wideband slope in (\ref{eq:min_S0}) are readily obtained.
\end{proof}

\section{Proof of Theorem \ref{theo:discrete2}}\label{APD_2}
\begin{proof}
From \cite{chang_book} and \cite{HARQ}, the LMGF of the arrival process and the effective capacity of the departure process are given, respectively, by
\begin{align}
\begin{cases}\label{eq:LMGF_1}
\Lambda_a(\theta)=\log_e\left(\frac{p_{11}+p_{22}e^{r\theta}+\sqrt{(p_{11}+p_{22}e^{r\theta})^2-4(p_{11}+p_{22}-1)e^{r\theta}}}{2}\right)\\
C_E(\tsnr)=\frac{R}{\mu}-\frac{R^2\sigma^2\theta}{2\mu^3}.
\end{cases}
\end{align}
Note that $C_E=-\frac{1}{\theta}\Lambda_c(-\theta)$ is a linear function of the LMGF of the departure process. Plugging the characterizations in (\ref{eq:LMGF_1}) into (\ref{eq:qos}), we obtain
\begin{align}\label{eq:qos2}
\frac{1}{2}\bigg(p_{11}+p_{22}e^{r\theta}+\sqrt{(p_{11}+p_{22}e^{r\theta})^2-4(p_{11}+p_{22}-1)e^{r\theta}}\bigg)\notag\\
\hspace{2cm}=e^{\theta C_E(\tsnr)}.
\end{align}
Then, by taking the derivative of both sides with respect to $\tsnr$ and evaluating as $\tsnr\rightarrow 0$, we have
\begin{equation}\label{eq:Eb_proof1}
\dot{r}(0)\theta\left(\frac{p_{22}}{2}+\frac{p_{22}(p_{11}+p_{22})-2(p_{11}+p_{22}-1)}{2(2-p_{11}-p_{22})}\right)=\theta\dot{C}_E(0).
\end{equation}
In determining (\ref{eq:Eb_proof1}), we have used the fact that $\lim_{\tsnr\rightarrow 0}r(\tsnr)=0$ and $\lim_{\tsnr\rightarrow 0}C_E(\tsnr)=0$. Note that when the transmit power approaches $0$, the departure rate should also go to $0$, which in turn makes the effective capacity approach $0$. To satisfy the queuing constraints, the arrival rate $r$ in the ON state should also diminish to $0$. From (\ref{eq:Eb_proof1}), we get
\begin{align}
\dot{r}(0)=&\dot{C}_E(0)\bigg/\left[\frac{p_{22}}{2}+\frac{p_{22}(p_{11}+p_{22})-2(p_{11}+p_{22}-1)}{2(2-p_{11}-p_{22})}\right] \\
         =&\dot{C}_E(0)/P_{ON}.
\end{align}
In the proof of Theorem \ref{theo:discrete}, we have shown that $\dot{C}_E(0)=\frac{\f_M(\varepsilon)}{\mu \log_e 2}$. Therefore, we can have the first order derivative of the throughput evaluated as $\tsnr$ goes to $0$ as
\begin{align}
\dot{r}_{\text{avg}}(0)=&\dot{r}(0) P_{ON} \\
                =&\dot{C}_E(0)\label{eq:Eb_proof8}\\
                =&\frac{\f_M(\varepsilon)}{\mu \log_e 2}.\label{eq:Eb_proof2}
\end{align}
Similarly, by taking the second order derivatives of both sides of (\ref{eq:qos2}) with respect to $\tsnr$ and evaluating as $\tsnr\rightarrow 0$, we obtain
\begin{align}
\ddot{r}(0)=\frac{\theta\dot{C}_E(0)^2+\ddot{C}_E(0)-\dot{C}_E(0)^2\theta(\zeta+1)}{P_{ON}}
\end{align}
where $\zeta$ is defined in (\ref{eta}).
In the proof of Theorem \ref{theo:discrete}, we show that $\ddot{C}_E(0)=-\frac{\f_M(\varepsilon)^2(\theta\sigma^2+\mu^2\log_e 2)}{\mu^3(\log_e 2)^2}$. Therefore, we can find
\begin{align}
\ddot{r}_{\text{avg}}(0)=&\ddot{r}(0) P_{ON}\\
                 =&\frac{\f_M(\varepsilon)^2(-\theta\zeta-\frac{\theta\sigma^2+\mu^2\log_e 2}{\mu})}{(\mu\log_e 2)^2}.\label{eq:Eb_proof3}
\end{align}
Inserting the results in (\ref{eq:Eb_proof2}) and (\ref{eq:Eb_proof3}) into (\ref{bitenergy}) and (\ref{widebandslop}), we get the desired results shown in Theorem \ref{theo:discrete2}.
\end{proof}

\section{Proof of Theorem \ref{theo:discrete4}}\label{APD_3}
\begin{proof}
The proof is similar to the proof of Theorem \ref{theo:discrete2}. From \cite{EB_Markov}, the LMGF of the arrival process of the ON-OFF Markov fluid source is given by
\begin{align}\label{eq:Eb_proof4}
\Lambda_a(\theta)=\frac{1}{2}\bigg(\theta r-\alpha-\beta+\sqrt{(\theta r-\alpha-\beta)^2+4\alpha\theta r}\bigg).
\end{align}
Plugging (\ref{eq:Eb_proof4}) into (\ref{eq:qos}), taking the first and second order derivatives and evaluating as $\tsnr\rightarrow 0$, we get
\begin{align}
\dot{r}(0)=\dot{C}_E(0)/P_{ON},
\end{align}
Using $\dot{r}(0)$, we get $\dot{r}_{\text{avg}}(0)$ as
\begin{align}
\dot{r}_{\text{avg}}(0)=&\dot{r}(0) P_{ON}\notag\\
                =&\dot{C}_E(0) \label{eq:Eb_proof9}\\
                =&\frac{\f_M(\varepsilon)}{\mu\log_e 2}.\label{eq:Eb_proof5}
\end{align}
Furthermore, we have
\begin{align}
\ddot{r}_{\text{avg}}(0)=&\ddot{r}(0)\frac{\alpha}{\alpha+\beta}\\
=&\ddot{C}_E(0)-\dot{C}_E^2(0)\theta\frac{2\beta}{\alpha(\alpha+\beta)}\\
                 =&-\left(\frac{\f_M(\varepsilon)}{\mu\log_e 2}\right)^2\left(\frac{\theta\sigma^2+\mu^2\log_e 2}{\mu}+\frac{2\theta\beta}{\alpha(\alpha+\beta)}\right).\label{eq:Eb_proof6}
\end{align}
Inserting the results in (\ref{eq:Eb_proof5}) and (\ref{eq:Eb_proof6}) into (\ref{bitenergy}) and (\ref{widebandslop}), we obtain the desired results in Theorem \ref{theo:discrete4}.
\end{proof}

\section{Proof of Theorem \ref{theo:MMPS1}}\label{APD_4}
\begin{proof}
From \cite{EB_Markov}, the LMGF of the arrival process of the ON-OFF MMPS is given by
\begin{align}\label{eq:Eb_MMPS1}
\Lambda_a(\theta)=&\frac{1}{2}\left[(e^\theta-1)\nu-(\alpha+\beta)\right]\notag\\
                  &\;+\frac{1}{2}\sqrt{\left[(e^\theta-1)\nu-(\alpha+\beta)\right]^2+4\alpha\nu(e^\theta-1)}.
\end{align}
Plugging (\ref{eq:Eb_MMPS1}) into (\ref{eq:qos}), we can find \begin{equation}\label{eq:MMPS_lmd}
\nu(\tSNR)=\frac{\theta\left[\theta C_E(\tsnr)+\alpha+\beta\right]}{(e^\theta-1)\left[\theta C_E(\tsnr)+\alpha\right]}C_E(\tsnr).
\end{equation}
Inserting (\ref{eq:MMPS_lmd}) into (\ref{eq:ravg_MMPS}), the throughput can be expressed as
\begin{align}\label{eq:ravg2_MMPS}
r_{\text{avg}}(\tsnr)=&\nu\:P_{ON}\notag\\
=&P_{ON}\frac{\theta\left[\theta C_E(\tsnr)+\alpha+\beta\right]}{(e^\theta-1)\left[\theta C_E(\tsnr)+\alpha\right]}C_E(\tsnr).
\end{align}
Taking the first and second order derivatives and evaluating as $\tsnr\rightarrow 0$, we get
\begin{align}
\dot{r}_{\text{avg}}(0)=&P_{ON}\frac{\theta\alpha(\alpha+\beta)}{\alpha^2(e^\theta-1)}\dot{C}_E(0)\\
          =&\frac{\theta}{e^\theta-1}\dot{C}_E(0), \label{eq:MMPS_proof1}
\end{align}
and
\begin{align}\label{eq:MMPS_proof2}
\ddot{r}_{\text{avg}}(0)=\frac{\theta}{e^\theta-1}\ddot{C}_E(0)-\frac{2\beta\theta^2}{(\alpha+\beta)(e^\theta-1)}\dot{C}_E^2(0)
\end{align}
where $\dot{C}_E(0)=\frac{\f_M(\varepsilon)}{\mu \log_e 2}$ and $\ddot{C}_E(0)=-\frac{\f_M(\varepsilon)^2(\theta\sigma^2+\mu^2\log_e 2)}{\mu^3(\log_e 2)^2}$. Inserting the results in (\ref{eq:MMPS_proof1}) and (\ref{eq:MMPS_proof2}) into (\ref{bitenergy}) and (\ref{widebandslop}), we obtain the desired results in Theorem \ref{theo:MMPS1}.
\end{proof}

\section{Proof of Theorem \ref{theo:discrete3}}\label{APD_5}
\begin{proof}
In order to determine the minimum energy per bit, we take the first derivative of $C_E(\tsnr)$ in (\ref{eq:C_E_opt_rate}) with respect to $\tSNR$ and express it as
\small
\begin{align}\label{eq:CE_der_opt_rate}
\hspace{-.4cm}\dot{C}_E(\tsnr)=\frac{\dot{R}^{\ast}(\tsnr)\mu(\tsnr)-\dot{\mu}(\tsnr)\R(\tsnr)}{(\mu(\tsnr))^2}-\frac{1}{4(\mu(\tsnr))^6} \bigg{[}2&(\mu(\tsnr))^3(2\R(\tsnr)\dot{R}^{\ast}(\tsnr)\sigma^2(\tsnr)\theta+(\R(\tsnr))^2\dot{\sigma^2}(\tsnr)\theta)\notag\\
                   &\hspace{2cm}-6(\mu(\tsnr))^2\dot{\mu}(\tsnr)(\R(\tsnr))^2\sigma^2(\tsnr)\theta\bigg{]},
\end{align}
\normalsize
where $\dot{R}^{\ast}(\tsnr)$, $\dot{\mu}(\tsnr)$ and $\dot{\sigma^2}(\tsnr)$ denote the first derivatives of $\R(\tsnr)$, $\mu(\tsnr)$, and $\sigma^2(\tsnr)$ with respect to $\tsnr$.
Next, we evaluate $\dot{C}_E(\tsnr)$ in (\ref{eq:CE_der_opt_rate}) at $\tsnr=0$. By exploiting the facts that $\R(\tsnr) \rightarrow 0$ when $\tsnr \rightarrow 0$, and $\dot{R}^{\ast}(0)=a$, and applying L'Hospital's Rule, we can express, after some simplifications, $\dot{C}_E(0)$ as
\begin{equation}\label{eq:Eb_proof7}
  \dot{C}_E(0)=\frac{a}{\mu(0)}.
\end{equation}
For the constant-rate arrival model, we have $r_{\text{avg}}=C_E$. Therefore, inserting the above result into (\ref{bitenergy}) provides the minimum energy per bit expression in (\ref{eq:EbN0_min_opt_rate}).
\end{proof}

\section{Proof of Theorem \ref{theo:discrete5}}\label{APD_6}
\begin{proof}
For ON-OFF dicrete Markov and Markov fluid sources, it is very easy to verify that (\ref{eq:Eb_proof8}) and (\ref{eq:Eb_proof9}) are still valid, and the only step we need to perform is to insert (\ref{eq:Eb_proof7}) into (\ref{eq:Eb_proof8}) and (\ref{eq:Eb_proof9}), which gives
\begin{align}
\dot{r}_{\text{avg}}(0)=\frac{a}{\mu(0)}
\end{align}
for both discrete and fluid ON-OFF Markov sources. Inserting this expression into (\ref{bitenergy}) proves (\ref{eq:EbN0_min_opt_rate_Markov}).

We can also verify that (\ref{eq:MMPS_proof1}) is still valid for MMPS. Inserting (\ref{eq:Eb_proof7}) into (\ref{eq:MMPS_proof1}), we obtain
\begin{align}
\dot{r}_{\text{avg}}(0)=\frac{\theta}{e^\theta-1}\frac{a}{\mu(0)}
\end{align}
and
\begin{align}
\ben_{\rm{min}}=&\frac{1}{\dot{r}_{\text{avg}}(0)}\\
=&\frac{e^\theta-1}{\theta}\frac{\mu(0)}{a},
\end{align}
proving (\ref{eq:EbN0_min_opt_MMPS}).
\end{proof}
\end{appendices}

\end{spacing}

\bibliographystyle{ieeetr}
\bibliography{HARQ}

\end{document}